\documentclass[twoside,onecolumn,11pt]{article}

\usepackage[T1]{fontenc} % Use 8-bit encoding that has 256 glyphs
\usepackage{times}
\usepackage{microtype} % Slightly tweak font spacing for aesthetics

\usepackage[english]{babel} % Language hyphenation and typographical rules

\usepackage[hmarginratio=1:1,top=32mm,columnsep=20pt]{geometry} % Document margins
\usepackage[hang, small,labelfont=bf,up,textfont=it,up]{caption} % Custom captions under/above floats in tables or figures
\usepackage{booktabs} % Horizontal rules in tables

\usepackage{lettrine} % The lettrine is the first enlarged letter at the beginning of the text

\usepackage{enumitem} % Customized lists
\setlist[itemize]{noitemsep} % Make itemize lists more compact

\usepackage{abstract} % Allows abstract customization
\usepackage{titlesec} % Allows customization of titles
\renewcommand\thesection{\arabic{section}} % Roman numerals for the sections
\renewcommand\thesubsection{\thesection.\arabic{subsection}} % roman numerals for subsections
\titleformat{\section}[block]{\large\scshape\centering}{\thesection.}{1em}{} % Change the look of the section titles
\titleformat{\subsection}[block]{\large}{\thesubsection.}{1em}{} % Change the look of the section titles

\usepackage{titling} % Customizing the title section

\usepackage{hyperref} % For hyperlinks in the PDF
\usepackage{amsfonts, amsmath, amsthm, amssymb} % For math fonts, symbols and environments
\usepackage{graphicx} % Required for including images
\usepackage{natbib}
\usepackage{bbm}
\usepackage{authblk}

\newtheorem{remark}{Remark}
\newtheorem*{defini}{Definition}
\usepackage{thmtools, thm-restate}

\pretitle{\begin{center}\Huge\bfseries} % Article title formatting
\posttitle{\end{center}} % Article title closing formatting
\title{Sparsity estimation in compressive sensing with application to MR images} % Article title

\author[1]{Jianfeng Wang\thanks{corresponding author with email: jianfeng.wang@umu.se}}
\author[1]{Zhiyong Zhou}
\author[2]{Anders Garpebring}
\author[1]{Jun Yu}
\affil[1]{Department of mathematics and mathematical statistics, Ume\aa{ } University}
\affil[2]{Department of radiation Science, Ume\aa{ } University}

\date{} % Leave empty to omit a date
\renewcommand{\maketitlehookd}{%
\begin{abstract}
\noindent
The theory of compressive sensing (CS) asserts that an unknown signal $\mathbf{x} \in \mathbb{C}^N$ can be accurately recovered from $m$ measurements with $m\ll N$
provided that $\mathbf{x}$ is sparse. Most of the recovery algorithms need the sparsity $s=\lVert\mathbf{x}\rVert_0$ as an input. However, generally $s$ is unknown,
and directly estimating the sparsity has been an open problem. In this study, an estimator of sparsity is proposed by using Bayesian hierarchical model. Its statistical properties such as unbiasedness and asymptotic normality are proved. In the simulation study and real data study, magnetic resonance image data is used as input signal, which becomes sparse after sparsified transformation. The results from the simulation study confirm the theoretical properties of the estimator. In practice, the estimate from a real MR image can be used for recovering future MR images under the framework of CS if they are believed to have the same sparsity level after sparsification. \\\\

\noindent Key words: Compressive sensing, Sparsity, Bayesian hierarchical model, Mat\'{e}rn covariance, MRI
\end{abstract}
}

\begin{document}

% Print the title
\maketitle

%----------------------------------------------------------------------------------------
%	ARTICLE CONTENTS
%----------------------------------------------------------------------------------------

\section{Introduction}
Compressive sensing (CS) initially emerged around the year 2006 \citep{Donoho2006,Candes2006}. The aim of CS is to recover a $s$-sparse signal $\mathbf{x}\in \mathbb{C}^N$ from $m$ noisy observations $\mathbf{y} \in \mathbb{C}^m$:
\begin{align*}
\mathbf{y}=A\mathbf{x}+\mathbf{e} \tag{1}
\end{align*}
In (1), $\mathbf{x}$ has $s$ non-zero elements implying that $s=\lVert\mathbf{x}\rVert_0$, $A\in \mathbb{C}^{m\times N}$ is a measurement matrix with $m\ll N$ satisfying restricted isometry property \citep[chapter 6]{Foucart2013} and $\mathbf{e}\in\mathbb{C}^m$ is additive noise such that $\lVert\mathbf{e}\rVert_2\le \eta$ for some $\eta\ge0$.

To recover $\mathbf{x}$ in $(1)$, it can be translated into a quadratically constrained $\ell_1$-minimization problem:
\begin{align*}
\underset{\mathbf{z}\in\mathbb{C}^N}{\text{minimize}}\quad \lVert\mathbf{z}\rVert_1 \quad \quad \text{subject to} \quad \lVert A\mathbf{z}-\mathbf{y}\rVert_2\le \eta.
\end{align*}
There are several specific algorithms to solve the optimization problem, e.g. orthogonal matching pursuit (OMP), compressive sampling matching pursuit (CoSaMP), iterative hard thresholding (IHT) and hard thresholding pursuit (HTP) \citep[chapter 3]{Foucart2013}. All the algorithms mentioned above need sparsity $s$ as an input. However, $s$ is typically unknown in practice. The difficulty to recover $\mathbf{x}$ consists in estimating the sparsity $s$.

In this study, we propose a novel estimator for sparsity by using Bayesian hierarchical model (BHM). By assuming the sparsity $s$ to be random, our target parameter is the mathematical expectation of $s$, $E(s)$. The assumption of the randomness of $s$ is reasonable in practice, since noise commonly exists in the process of acquiring signals such that it is impossible to obtain two sparse signals that are exactly the same even under the same control \citep{Henkelman1985}.
We estimate $E(s)$ by using an observed 2D sparse image $\mathbf{x}$. The unbiasedness of the estimator is derived analytically, and it can be confirmed through a simulation study. Another interesting finding is that the estimator is asymptotically normally distributed under regular conditions, which could be used to construct confidence interval of $E(s)$. This property is also confirmed through the simulation study. In the simulation study, 2D magnetic resonance (MR) image is considered as an input signal. MR images are not sparse in general, but they are compressible \citep{Haldar2010}.  For instance, they can be be transformed into a sparse image, $\mathbf{x}$, through wavelet thresholding techniques \citep{Prinosil2010}. In MR imaging circumstance, the measurement matrix $A$ is a partial Fourier matrix and $\mathbf{y}$ is the measurements in frequency domain called $k$-space. In practice, once the estimate from the model is obtained, it could be directly used for recovering future MR images under the framework of CS if they are believed to have the same sparsity level $(E(s))$ after sparsification, for example but not limited to two scans of one person's brain.

This paper is organized as follows. The statistical theory is introduced in Section 2. In Section 3, the methods used for the simulation study and real data analysis are specified. Results are presented in Section 4, followed by conclusion and discussion in Section 5. The paper is closed with proofs of the statistical properties in Appendix.

%%------------------------------------------------
%
\section{Theory}
Since $s=\lVert\mathbf{x}\rVert_0$, the concrete non-zero value of each element in $\mathbf{x}$ is not of interest. Assume that $s$ is random, in the sense of that each element of $\mathbf{x}$ is randomly assigned by either zero or non-zero value. Thus, instead of $s$, the mathematical expectation of $s$, $E(s)$, is the parameter of interest. In this section, we introduce a BHM \citep{Cressie2011} to construct a new estimator for $E(s)$ of an observed sparse image $\mathbf{x}$.

BHM is a statistical model consisting of multiple layers. It is common to have three layers in the model. The first layer is data model used to model observed data, the second layer is process model used to model the unknown parameters of interest in data model, and the last one is hyperparameter model used to model unknown hyperparameters. Let $o_i=\mathbbm{1}_{(x_i\ne 0)}, \text{ where } \mathbbm{1}$ is the indicator function. A Bernoulli distribution can be used to describe the assumption of randomness of $s$:
\begin{align*}
\text{Layer 1}\quad o_i|p_i\sim Ber(p_i),\quad 0<p_i<1
\end{align*}
Let $\mathbf{o}=\{o_1,o_2,...,o_N\}$ and $\mathbf{p}=\{p_1,p_2,...,p_N\}$, where $N$ is the number of elements in $\mathbf{o}$ as well as in $\mathbf{x}$.

Then a second layer is needed to describe the distribution of $p_i$ given by
\begin{align*}
\text{Layer 2}\quad \text{logit}(p_i)= \mu+m_i+\epsilon_i
\end{align*}
where $\mu$ is intercept term, $m_i$ represents structured spatial effect and $\epsilon_i$ represents unstructured spatial effect. It is very common to include these two types of spatial effects, since there is typically a correlated random structure that induces correlation based on neighborhood structure as well as an uncorrelated random noise that varies from pixel to pixel \citep{Carroll2015}. Let $\mathbf{m}=\{m_1,m_2,...,m_N\}$ be normal distributed with common mean $\mathbf{0}$ and precision matrix $Q(\boldsymbol{\theta})$, i.e. $\mathbf{m}\sim N(\mathbf{0}, Q(\boldsymbol{\theta}))$. The precision matrix is defined as the inverse of covariance matrix of the random field. Let $\boldsymbol{\epsilon}=\{\epsilon_1,\epsilon_2,...,\epsilon_N\}$ be independent and identically distributed normal with common mean $\mathbf{0}$ and precision matrix $\tau_{iid} \cdot\mathrm{I}_N$, i.e. $\boldsymbol{\epsilon}\sim N(\mathbf{0},\tau_{iid}\cdot \mathrm{I}_N)$, where $\tau_{iid}$ is marginal precision and $\mathrm{I}_N$ is an identity matrix with dimension $N\times N$. There are several spatial models using Markov property to construct the precision matrix $Q$ such that the random field $\mathbf{m}$ is called as Gaussian Markov Random Field (GMRF), for instance, Besag model \citep{Besag1991}, Leroux model \citep{Leroux2000} by using neighborhood information. The advantage of using GMRF is to produce great computational benefits \citep{Rue2005}. To use such models, conditional properties of the GMRF should be specified in advance, i.e. the order of neighborhood structure.
More often a GMRF with Mat\'{e}rn covariance is more intuitive and easier to specify for two distinct sites $i$ and $j$ than to specify the conditional properties, especially in geostatistics.
Mat\'{e}rn covariance is defined as \citep{Matern1960}:
\begin{align*}
Cov(m_i,m_j)=\frac{\sigma^2}{\Gamma(\nu)\times 2^{\nu-1}}(\kappa\lVert i-j\rVert_2)^\nu K_\nu(\kappa\lVert i-j\rVert_2),
\end{align*}
where $K_\nu$ is the modified Bessel function of the second kind, $\Gamma$ is the Gamma-function, $\nu$ is a smoothness parameter, $\kappa$ is a range parameter and $\sigma^2$ is marginal variance. Two important properties of a random field with Mat\'{e}rn covariance are that the random field is stationary and isotropic (when the distance is Euclidean) and the
covariance decreases as the distance between two sites $i$ and $j$ increases \citep{Matern1960,Stein1999}.
In this study, a GMRF with Mat\'{e}rn covariance is used. Actually, a continuous Gaussian random field with Mat\'{e}rn covariance function is a solution to stochastic partial
differential equation \citep{WHITTLE1954,WHITTLE1963},
\begin{align*}
(\kappa^2-\Delta)^{\alpha/2}z(\mathbf{s})&=\mathcal{W}(\mathbf{s}),\quad  \alpha=\nu+d/2,
\end{align*}
where $d$ is the dimension of the Gaussian random field $z(\mathbf{s})$ with location index $\mathbf{s}$, $\Delta$ is the Laplacian operator and $\mathcal{W}(\mathbf{s})$ is a
 random field with Gaussian white noises.
\citet{Lindgren2011} have proposed that a GMRF defined on a regular unit distance lattice with a sparse precision can be used
to represent a continuous Gaussian random field with Mat\'{e}rn covariance for $\nu\in \mathbb{Z}^+$. The sparseness of the precision matrix $Q$ is controlled by $\nu$. The smaller the $\nu$
is, the sparser the $Q$ is. For example, if $\nu=1$, $Q$ can be regarded as a precision matrix defined through the third order of neighborhood structure, and the non-zero values of $Q$ are controlled by $\kappa$.
\begin{align*}
\text{Layer 3}\quad &\text{It provides the prior distributions of the unknown hyperparameters such as}\\ &\mu, \boldsymbol{\theta}=\{\sigma, \kappa, \nu \} \text{ and } \tau_{iid}.
\end{align*}
Bayesian inference is applied to obtain the posterior distribution of $\mathbf{p}$, i.e. the distribution of $\mathbf{p}|\mathbf{o}$.

By using the posterior mean of $\mathbf{p}$, we construct an estimator for the mean sparsity $E(s)$:
\begin{align*}
\widehat{E(s)}=\sum_{i=1}^N E\left(p_i|\mathbf{O}\right),
\end{align*}
where $\mathbf{O}$ is a random field from which $\mathbf{o}$ is sampled. Its statistical properties are presented in the following propositions, of which the proofs are left in Appendix.
\begin{restatable}{proposition}{propone}
\label{prop1}
$\widehat{E(s)}$ is an unbiased estimator and its variance equals to the summation over all the elements of the covariance matrix of $E(\mathbf{p}|\mathbf{O})$.
\end{restatable}

Before presenting the next proposition, we introduce one definition and some notations which help to understand the proposition as well as its proof.
\begin{defini}
A set of random variables, of which the dimension could be any positive integer, is said to be $\rho$-radius dependent if any two random variables in the set are independent as long as the distance between the two random variables is greater than $\rho$, where $\rho\ge0$.
\end{defini}

\begin{remark}
 $\rho=0$ implies that the random variables are independent.
\end{remark}

\begin{remark}
 Under the model setting with Mat\'{e}rn covariance, $\{\text{logit}(p_i), 1\le i\le N\}$ are $\rho$-radius dependent random variables $(\rho>0)$, where the smoothing parameter $\nu$ is related to the spatial dependence $\rho$. For example, if $\nu=1$, $\rho=2$ and if $\nu=2$, $\rho=3$.
\end{remark}

Let $\rho^*=\left \lceil{\rho}\right \rceil$, the smallest integer greater than or equal to $\rho$.
Let $\phi$ be a positive integer greater than $\rho^*$. Let $n_1, n_2$ be the dimensional size of a sparse image such that $n_1 \times n_2 = N$. And the sparse image can be divided into a set of independent squares and borders which separate the squares. Each square has dimension $(2\phi+1)\times (2\phi+1)$ and consists of $(2\phi+1)^2$ random variables (pixels). The width of each border is $\rho^*$ and the border regions surrounding each square consist of $2(2\phi+1)\rho^*+(\rho^*)^2$ random variables. Let $n_{sq}$ be the number of squares. Let $S_k$ be the sum of the random variables in the $k$th square, $S_k^B$ be the sum of the random variables in the borders surrounding the $k$th square, where $k=1,2,...,n_{sq}$, also $\sigma^2_k$ be the variance of $S_k$ and $r^3_k=E|S_k|^3.$
\begin{restatable}{proposition}{proptwo}
\label{prop2}
$\widehat{E(s)}$ is asymptotically normally distributed as $n_1, n_2 \to \infty$, if
\begin{enumerate}[label=\alph*)]
  \item $E\left(p_i|\mathbf{O}\right)$'s are absolutely continuous random variables for $1\le i \le N$, and
  \item $n_{sq} \to \infty$, and $\phi \to\infty$ at a rate slower than $n_{sq}^{1/6}$.
\end{enumerate}
\end{restatable}

\begin{remark}
Condition $b)$ says that one can choose a relatively smaller smoothness parameter $\nu$ compared to the size of the image, implying that the spatial dependence is not strong, otherwise this assumption may not hold.
\end{remark}

\begin{remark}
 The asymptotics in the proposition refers to increasing-domain asymptotics. There is another type of asymptotics, i.e. infill asymptotics (fixed-domain asymptotics). The infill asymptotics can lead the spatial dependence $\rho$ to increase rapidly as well as $\phi$, which might break down condition $b)$. This is in line with what \citet{Cressie1993} mentioned, infill asymptotics is preferable in geostatistical data, whereas increasing-domain asymptotics is often more appropriate for lattice data.
\end{remark}
%%------------------------------------------------
\section{Methods}
In this section, we briefly introduce background about MR images to be used and specify how to sparsify the images such that it can be analysed using the BHM described in Section 2. Furthermore, how to set the prior distributions for the BHM and how to evaluate the performance of the estimator are also presented in this section.
\subsection{Input data}
Two kinds of MR images are analysed in the study. One kind is simulated brain images with resolution $128\times 128$, and the other kind is a real brain image with resolution $256\times 256$. The simulated images were produced by using simulated gradient echo based dynamic contrast enhanced MRI scans at 3T with a noise level equivalent of a 12-channel head coil  \citep[see][for details]{Brynolfsson2014}. The real brain image was acquired with a 2D spin-echo sequence on a 3T GE Signa Positron emission tomography (PET)/MR scanner.

To be able to fit the BHM to a sparsified MR image, the following question should be answered. Is the sparsity of sparsified MR image random? Assume $\mathbf{x}_1,\mathbf{x}_2$ are two sparse images transformed from two sequential scans of MR images, $\mathbf{x}_{MRI_1}\text{ and } \mathbf{x}_{MRI_2}$, of the same brain under the same conditions, any slight difference between the two images may result in the positions as well as the numbers of non-zero values in $\mathbf{x}_1 \text{ and } \mathbf{x}_2$ are different, i.e. $s_1 \ne s_2$, which implies that $s$ is varying and can be considered as random.
\subsection{Sparsification}
Since MR image is not sparse in general, one discrete wavelet transform (DWT) followed by one thresholding method could be used to transform a MR image to a sparse image in wavelet domain. There are many DWT and thresholding methods can be used \citep{Vanithamani2011}. Since this study does not focus on evaluation of the performance among these methods, DWT of Daubechies 1 with level 3 is used and followed by the hard thresholding method to eliminate the noise. The reason of using wavelet thresholding method is that wavelet transform is good at energy compaction, the smaller coefficients represent noise and larger coefficients represent the image features. DWT decomposes the image into four sub-bands in general (i.e. A, DH,  DV and DD) shown in Figure 1.
\begin{figure}
    \centering
    \includegraphics[width=0.4\linewidth,keepaspectratio]{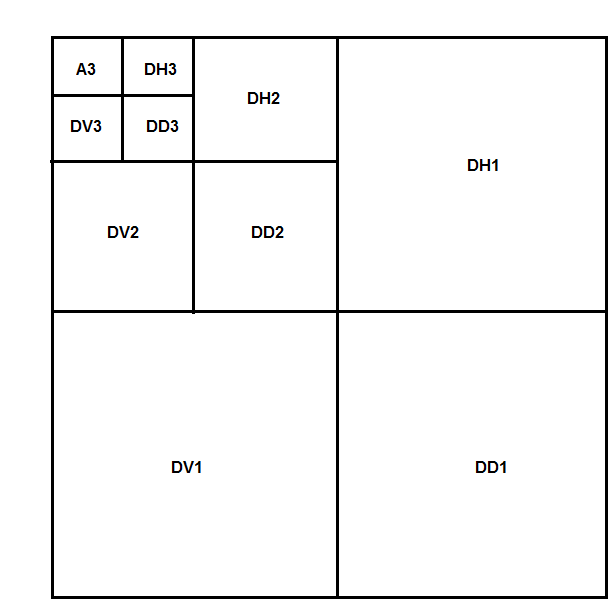}
    \caption{\small {Discrete Wavelet Transform with three levels}}
    \label{fig:2}
\end{figure}
The numbers $1,2 \text{ and } 3$ in Figure 1 indicate the three levels of the DWT, while DH, DV and DD are the detail sub-bands in horizontal, vertical and diagonal directions respectively and A sub-band is the approximation sub-band.

The hard thresholding method eliminates coefficients that are smaller than a certain threshold value $T$. The function of hard thresholding is given by
\[
    f(c)=
\begin{cases}
    c,& \text{if } |c|\ge T\\
    0,              & \text{otherwise}
\end{cases}
\]
where $c$ is the detailed wavelet coefficient. Note that when the threshold value $T$ is too small, the noise reduction is not sufficient. In the other way around, when $T$ is too large, the noise reduction is over sufficient. In this study, one of the most commonly used method is considered to estimate the value $T$ \citep{Braunisch,Prinosil2010,Vanithamani2011}, and the estimator is defined as
\begin{align*}
\hat{T}=\sigma_{image}\sqrt{2 \log(N)}
\end{align*}
where $N$ is the number of image pixels, $\sigma_{image}$ is the standard deviation of noise of the image which can be estimated by \citep{Donoho1995,Prinosil2010,Vanithamani2011}
\begin{align*}
\hat{\sigma}_{image}=\frac{\text{median}|c_D^1|}{0.6745}
\end{align*}
where $c_D^1$ indicates detailed wavelet coefficients from level 1.

\subsection{Priors of the BHM}
After thresholding, a sparse image in wavelet domain is obtained. Before fitting the BHM to the sparse image, prior distributions of the unknown random variables in the third layer should be specified. A flat prior is assigned to $\mu$ which is equivalent to a Gaussian distribution with $0$ precision. The priors for $\log(\tau_{iid})$ and $\log(1/\sigma^2)$ are set to be Log-Gamma$(1,5\times 10^{-5})$. The Log-Gamma distribution is defined as that a random variable $X\sim$ Log-Gamma$(a, b)$ if $\exp(X)\sim$ Gamma$(a, b)$ \citep{SaraMartino2010}. Thus, a Log-Gamma$(a, b)$ distribution is assigned to the logarithm of the precision, e.g. $\log(\tau_{iid})$, which is equivalent to assign a Gamma$(a, b)$ to $\tau_{iid}$, and the prior knowledge about $\tau_{iid}$ is reflected through $a$ and $b$. The prior for $\log(\sqrt{8\nu}/\kappa)$ is set to be Log-Gamma$(1,10^{-2})$. $\sqrt{8\nu}/\kappa$ represents a distance where the covariance is about $0.1$.

$\nu$ is treated as a fixed number. In this study, $\nu=1$ is used, which implies that for a pixel in the sparse image, the conditional mean of this pixel given remaining pixels is only affected by its third order nearest neighbors.

R-INLA \citep{Rue2009} is used to implement the BHM.

\subsection{Evaluation}
$E(s)$ is the parameter of interest, which can also be estimated by another unbiased estimator, i.e. the sample mean $\widehat{E_{sim}(s)}=\frac{1}{n}\sum_{i=1}^ns_i$, where $n$ is the number of simulated images and $s_i$ is the sparsity of the $i$th sparse image. We use $\widehat{E_{sim}(s)}$ as a reference of the true mean value by choosing a larger value for $n$ according to law of large numbers (LLN). A series of $\widehat{E(s)}_I$ can be obtained from the simulated images, which can be used to compare with $\widehat{E_{sim}(s)}$ in order to confirm the theoretical property of the unbiasedness. The unbiasedness is measured by the mean of absolute difference in percentage $|\widehat{E_{sim}(s)}-\widehat{E(s)}_I|*100/N$ over the series of $\widehat{E(s)}_I$. The range of $I$ should not be too small in terms of LLN, whereas it should not be too large in terms of computational time. To state that the estimator is unbiased, the measurement should be close to zero. Besides the unbiasedness, the measurement also indicates the difference in terms of image size. Furthermore, the asymptotic normality of the BHM estimator could also be examined in the simulation study, since in this case $n_1=n_2=128$ (large dimensional size) and $\rho^*=2$ (weak spatial dependence), which indicates the condition $b)$ in Proposition 2 could be possibly met.

The evaluation methods mentioned above can not be extended to real images analysis, since it is not practical to scan one brain even for several times. We only compare $\widehat{E(s)}$ with the true sparsity and measure the absolute difference in percentage, i.e. $|\widehat{E(s)}-s|*100/N$.
\section{Results}
\subsection{Simulated MR images}
Figure 2 is one slice of simulated MR image of human brain with resolution $128\times 128$.
\begin{figure}
    \centering
    \includegraphics[width=0.3\linewidth,keepaspectratio]{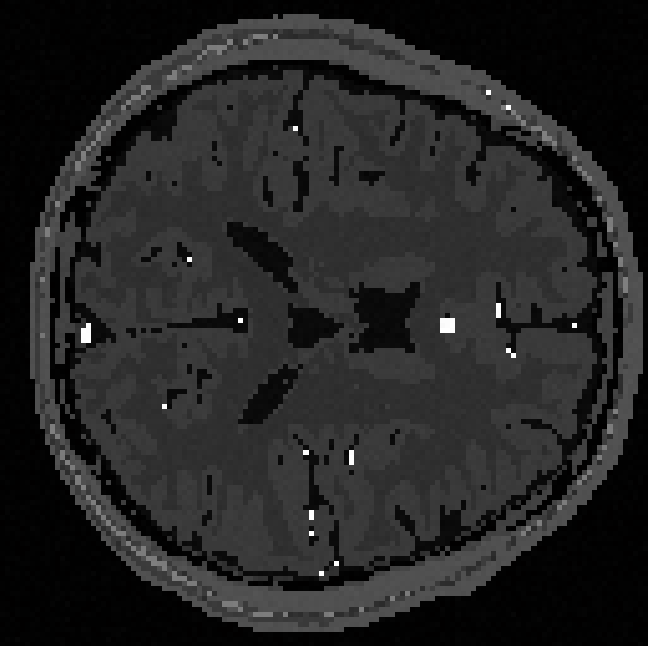}
    \caption{\small {A slice of simulated human brain with resolution $128\times 128$}}
    \label{fig:1}
\end{figure}
The DWT of Daubechies 1 with level 3 is shown in the left sub-figure of Figure 3.
\begin{figure}
    \centering
    \includegraphics[width=0.8\linewidth,keepaspectratio]{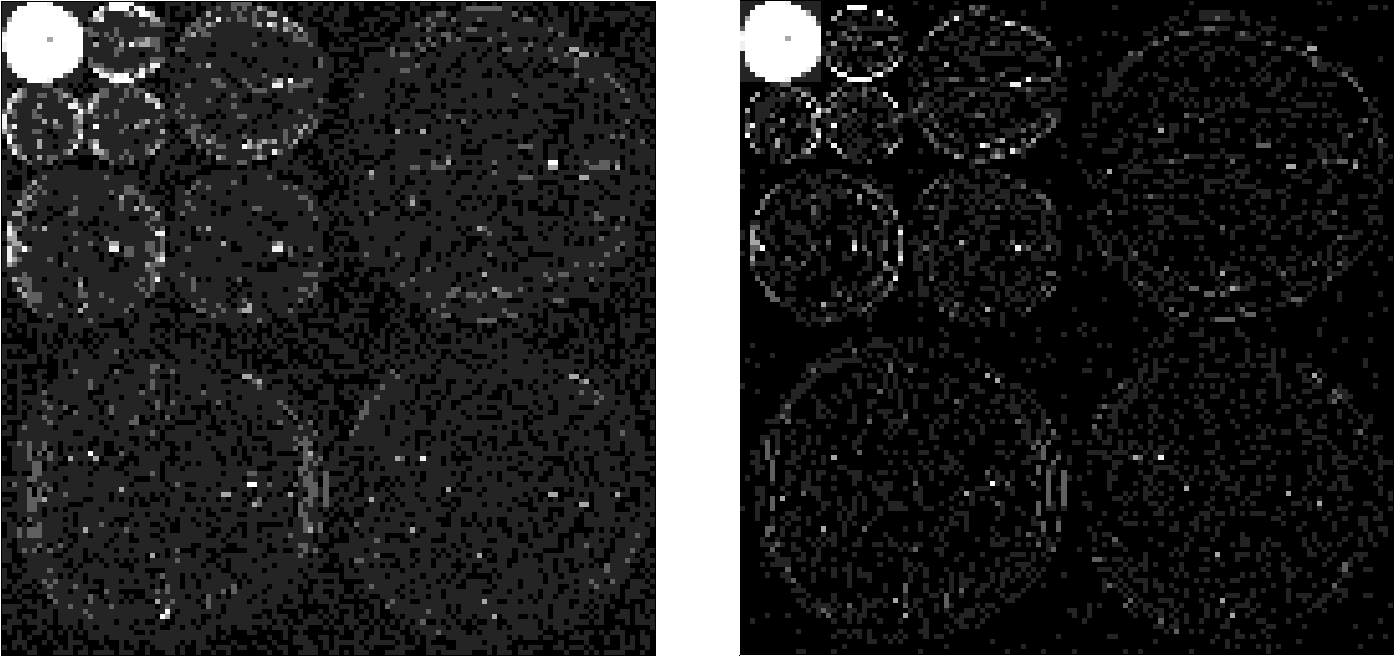}
    \caption{\small {Wavelet transformed image}}
    \caption*{Left: discrete wavelet transformed image. Right: the sparsified image. }
    \label{fig:3}
\end{figure}
The most upper left corner of the sub-figure is the approximation sub-band, while the remaining parts are the detail sub-bands. The estimated threshold value $\hat{T} \approx 0.01$, implying that the detailed coefficients which are less than 0.01 are set to $0$. The sparsified image is shown in the right sub-figure of Figure 3. It is difficult to see the difference between these two sub-figures except that the right sub-figure in general is darker than the left sub-figure, which is a consequence of thresholding. From the right sub-figure in Figure 3, it also shows that the non-zero coefficients are clustered, implying that given a pixel with higher probability to be a non-zero coefficient, its neighboring pixels should also have higher probabilities to have non-zero coefficients, and this relationship falls apart as the distance between two pixels becomes larger. This phenomenon could be described either by Mat\'{e}rn covariance or by its correspond sparse precision matrix. Thus, it confirms the reasonability of using BHM with Mat\'{e}rn covariance.

After fitting the BHM to the sparsified image, the posterior mean of $p_i$ for each pixel can be estimated and is shown in Figure 4.
\begin{figure}
    \centering
    \includegraphics[width=1.1\linewidth,keepaspectratio]{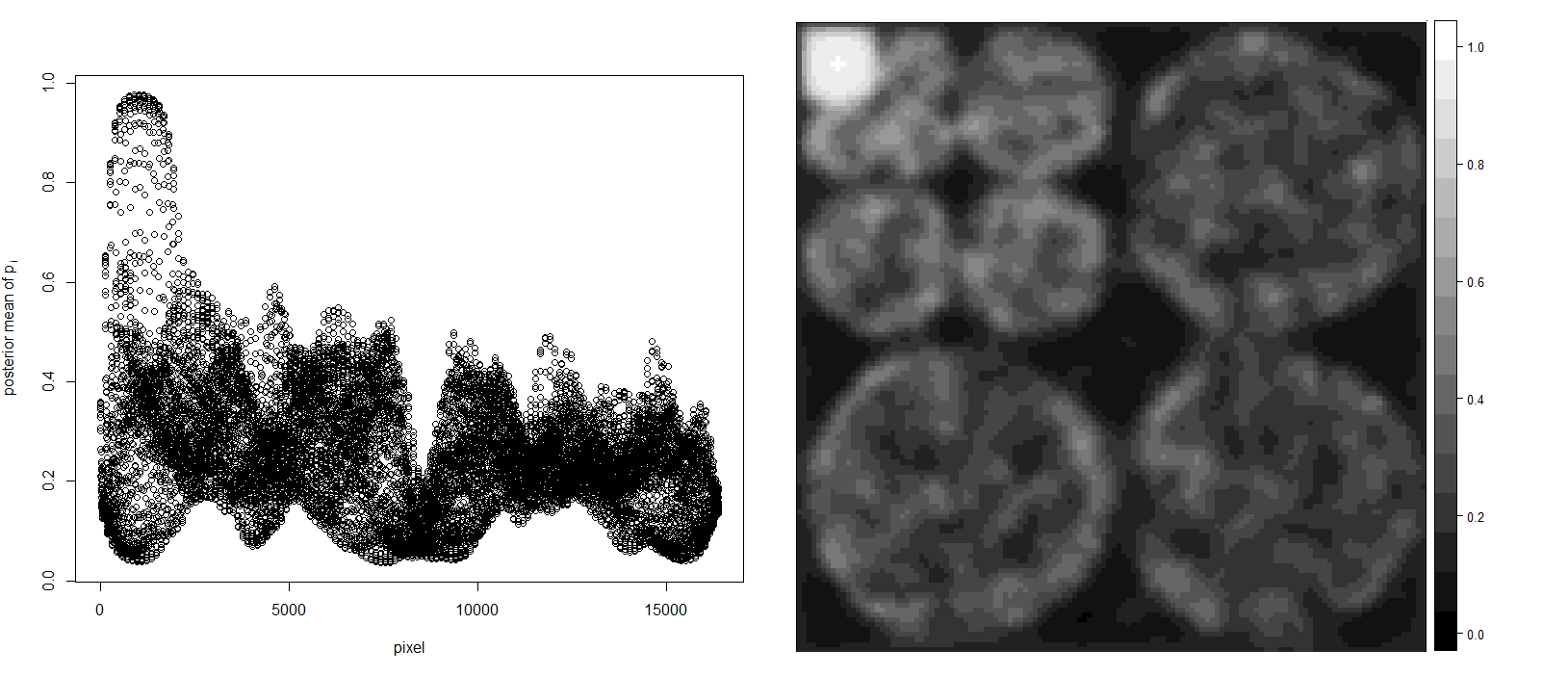}
    \caption{\small {Posterior mean of $p_i$ for every pixel}}
    \caption*{Left: scatter plot of the posterior mean of $p_i$. Right: image form of the posterior mean of $p_i$. }
    \label{fig:4}
\end{figure}
The left sub-figure of Figure 4 is scatter plot of the posterior mean of $p_i$, while the right sub-figure of Figure 4 shows the posterior mean of $p_i$ in image form which is easier to relate the posterior mean of $p_i$ to the sparse image in Figure 3. It shows some pixels, located at the most upper left corner of the right sub-figure of Figure 4, are with higher probabilities to have non-zero coefficients, while most of the remaining pixels are with lower probabilities to have non-zero coefficients. The summation of the posterior mean of $p_i$ over all pixels, i.e. the estimator of $E(s)$, is $4241.768$. The true sparsity of the sparsified MR image is $4213$. The absolute difference between the estimate and the true sparsity in percentage $|\widehat{E(s)}-s|*100/N\approx 0.2$. Afterwards, 1000 simulated MR images were generated under the same settings as the one in Figure 2. $\widehat{E_{sim}(s)}=4219.289$ and the absolute difference between $\widehat{E_{sim}(s)}$ and the estimate in percentage $|\widehat{E_{sim}(s)}-\widehat{E(s)}|*100/N\approx 0.1$.

By far, we only illustrate the performance of the BHM estimator from a single image. Further, to evaluate the unbiasedness and stability of the estimator, 90 $\widehat{E(s)}$'s out of the 1000 simulations are calculated. The scatter plot of the absolute difference in percentage $|\widehat{E_{sim}(s)}-\widehat{E(s)}_I|*100/N$ for the 90 simulations, where $I=1,2,...,90$, is shown in Figure 5.
\begin{figure}
    \centering
    \includegraphics[width=0.7\linewidth,keepaspectratio]{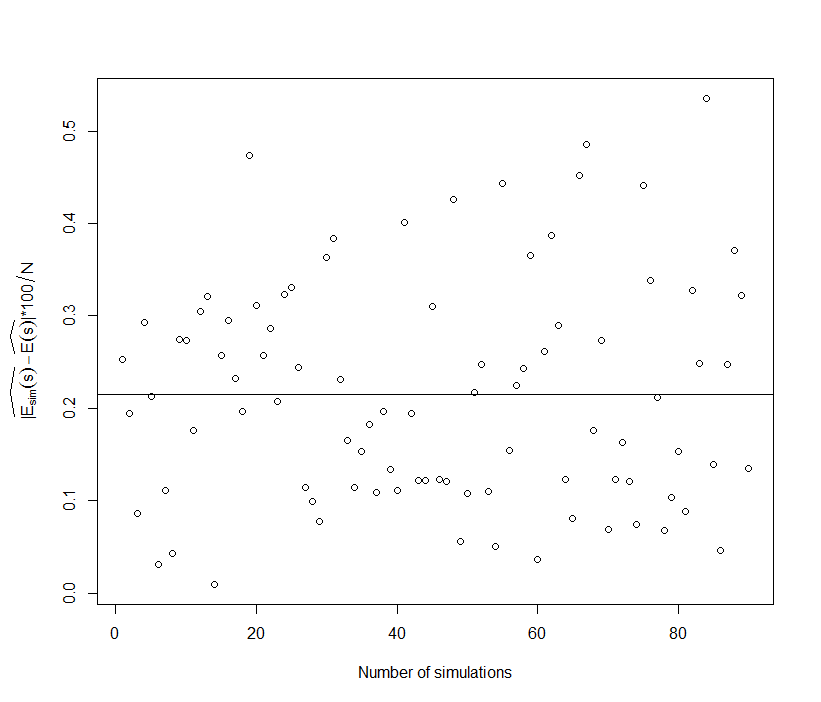}
    \caption{\small {$|\widehat{E_{sim}(s)}-\widehat{E(s)}|*100/N$ for 90 simulations}}
    \label{fig:5}
\end{figure}
The black line is the mean of $|\widehat{E_{sim}(s)}-\widehat{E(s)}_I|*100/N$ over the 90 simulations, and its value is about $0.21$, which is a relatively small number that could confirm the theoretical property of unbiasedness. From Figure 5, $|\widehat{E_{sim}(s)}-\widehat{E(s)}|*100/N\in [0.01,0.53]$. The variance of $\widehat{E(s)}$ based on the 90 simulations is $480.069$. All of these show that the BHM estimator from an image does not diverge too much from $\widehat{E_{sim}(s)}$. Furthermore, a normal quantile-quantile plot of standardized $\widehat{E(s)}$ from the 90 simulations is shown in Figure 6.
\begin{figure}
    \centering
    \includegraphics[width=0.7\linewidth,keepaspectratio]{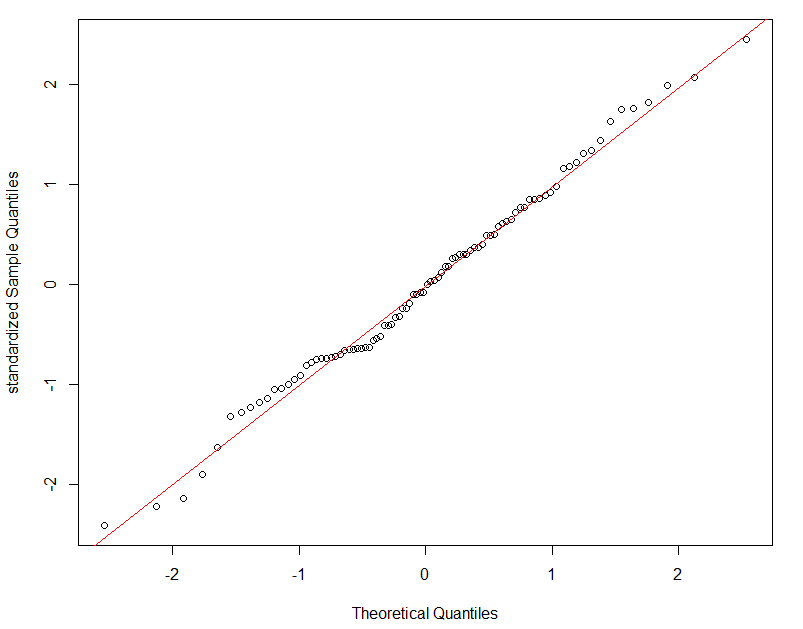}
    \caption{\small {Normal quantile-quantile plot of standardized $\widehat{E(s)}$ from 90 simulations}}
    \label{fig:9}
\end{figure}
The Shapiro-Wilk test of normality is performed and p-value $=0.7715$, implying that the estimator is normal distributed. It confirms the theoretical result about asymptotic normality of the estimator.

%%------------------------------------------------

\subsection{Real MR image}
The same sparsification procedure as for the simulated images is applied here. One slice of real MR image of human brain with resolution $256\times256$ is shown in Figure 7.
\begin{figure}
    \centering
    \includegraphics[width=0.3\linewidth,keepaspectratio]{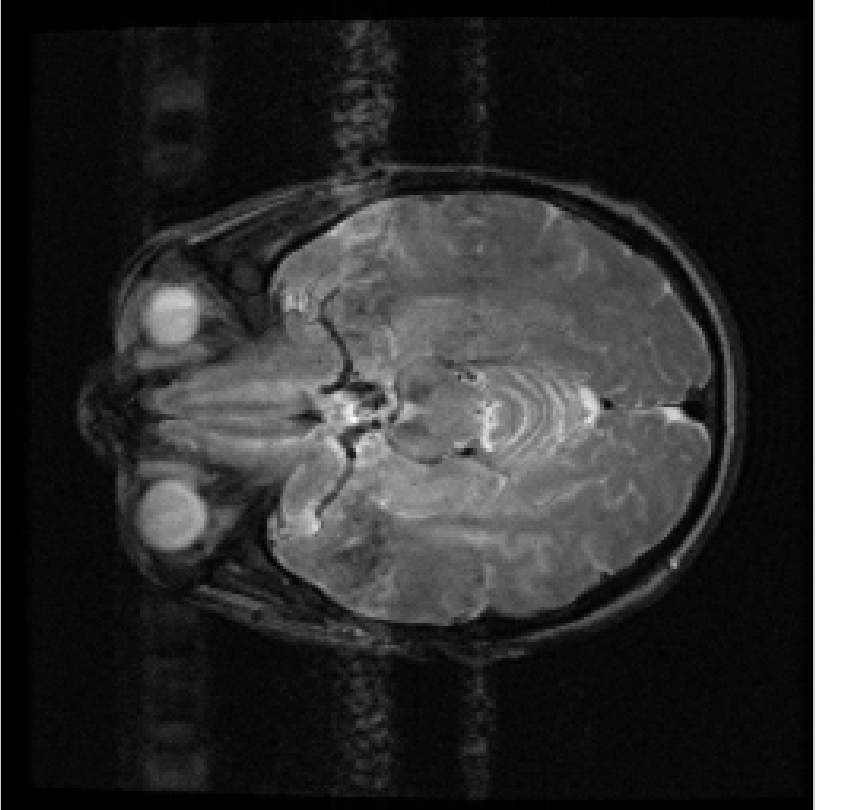}
    \caption{\small {One slice of real human brain with resolution $256\times 256$}}
    \label{fig:6}
\end{figure}
The DWT of Daubechies 1 with level 3 is shown in the left sub-figure of Figure 8. Sequentially, the hard thresholding method is used to eliminate the noise, and the result is shown in the right sub-figure of Figure 8.
\begin{figure}
    \centering
    \includegraphics[width=0.8\linewidth,keepaspectratio]{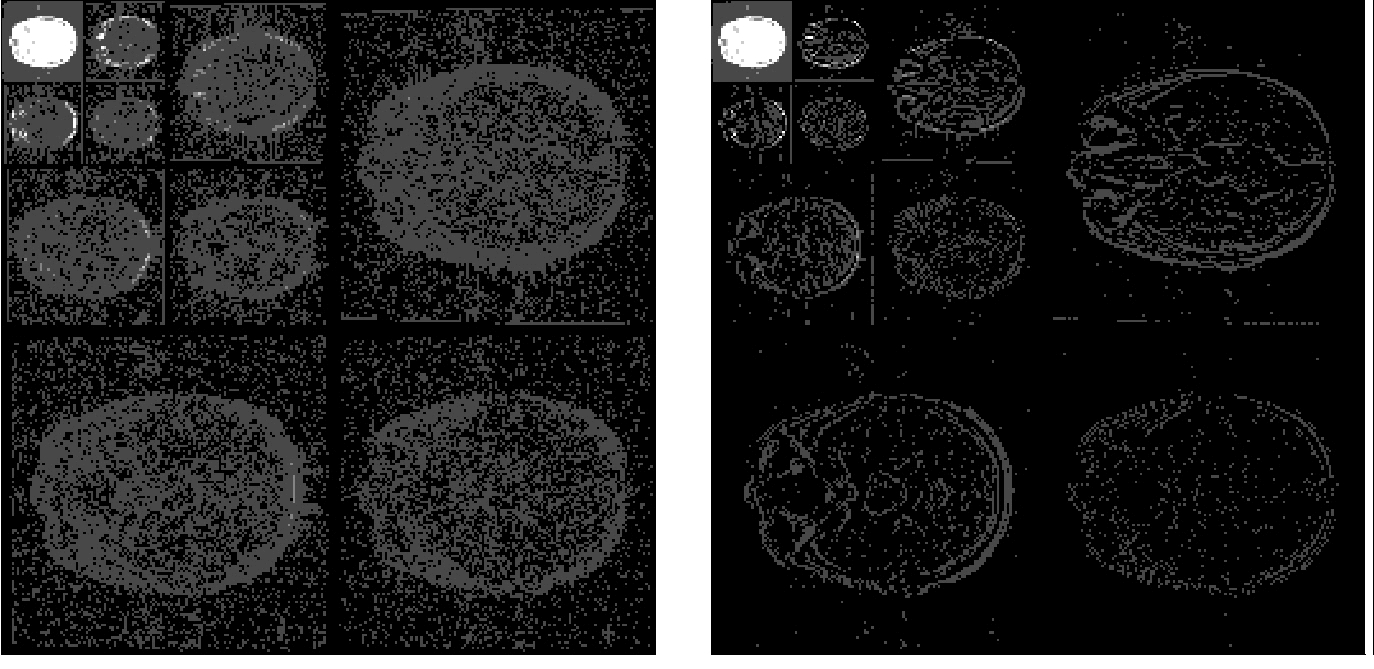}
    \caption{\small {Wavelet transformed image}}
    \caption*{Left: discrete wavelet transformed image. Right: the sparsified image.}
    \label{fig:7}
\end{figure}

By analysing the sparsified image using the BHM, the posterior mean of $p_i$ for each pixel is estimated and shown in Figure 9. The left sub-figure of Figure 9 is scatter plot of the posterior mean of $p_i$. The right sub-figure of Figure 9 is the posterior mean of $p_i$ in the image form. The summation of the posterior mean of $p_i$ over all pixels is $8254.679$. The true sparsity of the sparsified MR image is $8120$. The absolute difference in percentage $|\widehat{E(s)}-s|*100/N\approx 0.2$.
\begin{figure}
    \centering
    \includegraphics[width=1.1\linewidth,keepaspectratio]{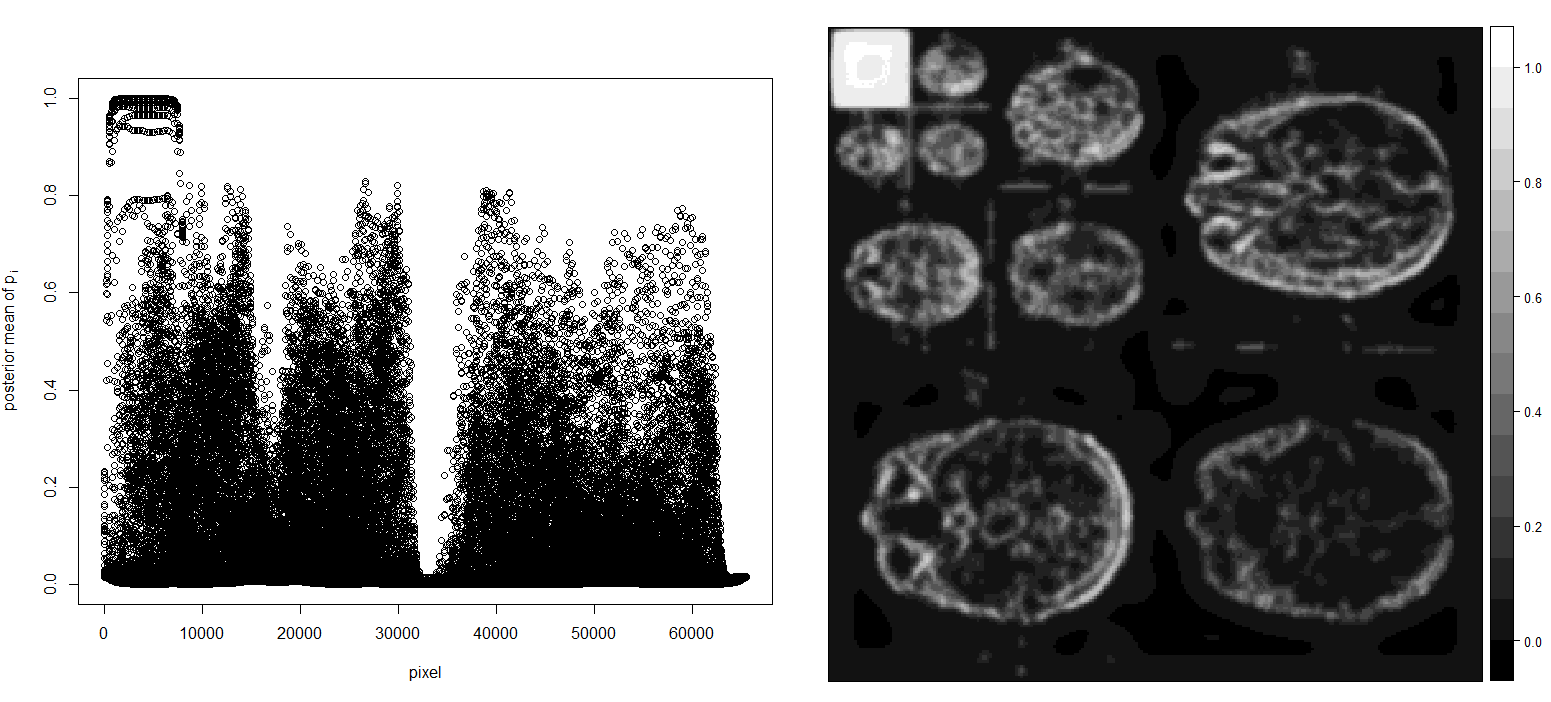}
    \caption{\small {Posterior mean of $p_i$ for every pixel}}
    \caption*{Left: scatter plot of the posterior mean of $p_i$. Right: image form of the posterior mean of $p_i$.}
    \label{fig:8}
\end{figure}
\section{Conclusion and Discussion}
In the study, we propose a novel estimator for the mathematical expectation of sparsity, $E(s)$, and prove that it is unbiased and asymptotically normally distributed. Its variance can also be derived analytically. Simulation study is used to confirm the theoretical results. The absolute difference in percentage, i.e. $|\widehat{E_{sim}(s)}-\widehat{E(s)}|*100/N$, is about $0.21$ in average, which indicates the unbiasedness. The asymptotic normality is also examined through the simulation study. The real data analysis is used for illustration of the new method's applicability in the sense of that $\widehat{E(s)}$ could be used directly in the recovery algorithms of CS for future MR images which are believed to have the same sparse level after sparsification as the one used in the study.

There are some issues that are not considered in this study. Relatively conservative priors are used in the study, thus more informative priors could be used according to properties of MR images. Also different models can be used to model the GMRF, e.g. Besag and Leroux model, and comparisons are made among different models. Besides these, it is possible to fit the model to a 3D image and prove that the statistical properties still hold for 3D images. Based on the image patterns of the posterior means shown in Figure 4 and 9, the dimension of measurement matrix, $A\in \mathbb{C}^{m\times N}$, is possible to be further reduced. Since if some pixels in $\mathbf{x}\in \mathbb{C}^N$ are believed to have zero coefficients, i.e. the probabilities are close to $0$, then these pixels do not need to be measured through CS. Similarly, if some pixels are believed to possess image features, i.e. the probabilities are close to $1$, then the values can be measured by a direct method rather than by CS. And CS is only applied to the remaining pixels, which leads to dimension reduction.

\section{Acknowledgements}
This work was supported by the Swedish Research Council grant [Reg.No.: 340-2013-5342].

\section*{Appendix}
\renewcommand{\thesubsection}{\Alph{subsection}}
\subsection{Proof of Proposition~\ref{prop1}}
\propone*
\begin{proof}[Proof:]
$ $\newline
The unbiasedness follows from the fact that
\begin{align*}
&E\left(\widehat{E(s)}\right)=E\left(\sum_{i=1}^N E\left(p_i|\mathbf{O}\right)\right)=\sum_{i=1}^N E(p_i)=E\left(\sum_{i=1}^N p_i\right)=E\left(\sum_{i=1}^N E\left(\mathbbm{1}_{(x_i\ne 0)}|p_i\right)\right)\\
&\qquad\qquad{ }
=\sum_{i=1}^N E\left(\mathbbm{1}_{(x_i\ne 0)}\right)= E\left( \sum_{i=1}^N \mathbbm{1}_{(x_i\ne 0)}\right)=E(s).\\
&\text{Calculation of variance is straightforward: }\\
&Var\left(\widehat{E(s)}\right)=Var\left(\sum_{i=1}^N E\left(p_i|\mathbf{O}\right)\right)
%=&E\left(\left(\sum_{i=1}^N E\left(p_i|\mathbf{O}\right)\right)^2\right)-E^2\left(\sum_{i=1}^N E\left(p_i|\mathbf{O}\right)\right)\\
%=&E\left(\sum_{i=1}^N E^2\left(p_i|\mathbf{O}\right)+2\sum_{1\le i<j\le N}E\left(p_i|\mathbf{O}\right)E\left(p_j|\mathbf{O}\right)\right)-\left(\sum_{i=1}^N E(p_i)\right)^2\\
%=&\sum_{i=1}^N \left(Var\left(E\left(p_i|\mathbf{O}\right)\right)+E^2\left(p_i\right)\right)+2\sum_{1\le i<j\le N}\left(Cov\left(E\left(p_i|\mathbf{O}\right), E\left(p_j|\mathbf{O}\right)\right)+E\left(p_i\right)E\left(p_j\right)\right)-\\
%&\left(\sum_{i=1}^N E^2\left(p_i\right)+2\sum_{1\le i<j\le N}E\left(p_i\right)E\left(p_j\right)\right)\\
=\sum_{i=1}^N Var\left(E\left(p_i|\mathbf{O}\right)\right)+2\sum_{1\le i<j\le N}Cov\left(E\left(p_i|\mathbf{O}\right), E\left(p_j|\mathbf{O}\right)\right).
\end{align*}
\end{proof}

\subsection{Proof of Proposition~\ref{prop2}}
\proptwo*
\begin{proof}[Proof:]
$ $\newline
We prove this proposition by verifying that all conditions in the theorem in \citet{Harvey2010} are met, since \citet{Harvey2010} has proven that the sum of $\rho$-radius dependent three dimensionally indexed random variables is asymptotically normally distributed under conditions. The proof for two dimensionally indexed random variables in our case follows similar to that given in \citet{Harvey2010} if
\begin{enumerate}[label=\roman*)]
  \item $\{E\left(p_i|\mathbf{O}\right): 1\le i \le N\}$ are $\rho$-radius dependent random variables, and
  \item $\sigma_k^2$ and $r_k^3$ are finite for a given square size, and
  \item $n_{sq} \to \infty$ and $\phi \to\infty$ at a rate slower than $n_1, n_2$ as well as $n_{sq}$ such that
  \begin{align*}
  \frac{\sum_k Var(S_k^B)}{n_{sq}\sum_k\sigma^2_k}&\to 0\tag{B1}\\
  \frac{(\sum_k r_k^3)^{1/3}}{(\sum_k\sigma_k^2)^{1/2}}&\to 0 \tag{B2}.
  \end{align*}
\end{enumerate}

We verify the three conditions one by one, sequentially. It has been shown that the dependence structure of a transformed GMRF is the same as the one of original GMRF if the transformed GMRF consists of absolutely continuous variables \citep{Prates2015,Cardin2009}. Since Mat\'{e}rn covariance is used to model $\text{logit}(\mathbf{p})$ under the model setting, in which the random variables are $\rho$-radius dependent, the same dependence structure is inherited by $E\left(\mathbf{p}|\mathbf{O}\right)$ if condition $a)$ holds, i.e. two distinct sites for $E\left(\mathbf{p}|\mathbf{O}\right)$ are pairwise non-negative correlated and $\{E\left(p_i|\mathbf{O}\right): 1\le i \le N\}$ are $\rho$-radius dependent random variables. Thus condition i$)$ is met.

By using the fact $0< E\left(p_i|\mathbf{O}\right)< 1 $ almost surely, it follows that
\begin{align*}
r_k^3&=E|S_k|^3=E(S_k)^3= E\left(\sum_jE\left(p_{kj}|\mathbf{O}\right)\right)^3< (2\phi+1)^6,
\end{align*}
where subscript $j=1,2,...,(2\phi+1)^2$, subscript $kj$ denotes the $j$th random variable in the $k$th square. Thus, $r_k^3$ is finite for a given $\phi$, so is $\sigma^2_k$, and the condition $ii)$ is met.

Next we verify $(B1)$ and $(B2)$ which are met under condition $b)$. To verify $(B1)$ and $(B2)$, we need use $Var(S_k^B)$ and $\sigma_k^2$ besides $r_k^3$ which has been shown above.

First, we calculate $Var(S_k^B)$. Since
\begin{align*}
Var\left(E\left(p_{kl}|\mathbf{O}\right)\right)<E\left(E\left(p_{kl}|\mathbf{O}\right)\right)^2<1,
\end{align*}

\begin{align*}
Cov\left(E\left(p_{kl}|\mathbf{O}\right),E\left(p_{kq}|\mathbf{O}\right)\right)&=E\left(E\left(p_{kl}|\mathbf{O}\right)E\left(p_{kq}|\mathbf{O}\right)\right)-E(p_{kl})E(p_{kq})\\
&<E\left(E\left(p_{kl}|\mathbf{O}\right)E\left(p_{kq}|\mathbf{O}\right)\right)\\
&<1,
\end{align*}
and $\phi>\rho^*\ge1$, it follows that
\begin{align*}
Var(S_k^B)&=Var\left(\sum_lE\left(p_{kl}|\mathbf{O}\right)\right)=\sum_lVar\left(E\left(p_{kl}|\mathbf{O}\right)\right)+\sum_{l\ne q}Cov\left(E\left(p_{kl}|\mathbf{O}\right),E\left(p_{kq}|\mathbf{O}\right)\right)\\
&<2(2\phi+1)\rho^*+(\rho^*)^2+[2(2\phi+1)\rho^*+(\rho^*)^2][2(2\phi+1)\rho^*+(\rho^*)^2-1] \\
&=4(\rho^*)^2(2\phi+1)^2+4(\rho^*)^3(2\phi+1)+(\rho^*)^4\\
&<C_1(2\phi+1)^4,
\end{align*}
where $l,q=1,2,...,2(2\phi+1)\rho^*+(\rho^*)^2$, $C_1$ is a positive constant that does not depend on $n_1,n_2$.

Then we calculate $\sigma_k^2$. Since $\{E\left(p_i|\mathbf{O}\right): 1\le i \le N\}$ are $\rho$-radius dependent random variables and any two random variables in the set are non-negative correlated, it follows that
\begin{align*}
\sigma_k^2&=Var\left(\sum_jE\left(p_{kj}|\mathbf{O}\right)\right)\ge\sum_jVar\left(E\left(p_{kj}|\mathbf{O}\right)\right)\ge C_2(2\phi+1)^2>0,
\end{align*}
where $C_2>0$ is the smallest variance of $E\left(p_{kj}|\mathbf{O}\right)$ and also does not depend on $n_1,n_2$.

Therefore,
\begin{align*}
\frac{\sum_k Var(S_k^B)}{n_{sq}\sum_k\sigma^2_k}&\le\frac{n_{sq}C_1(2\phi+1)^4}{n_{sq}^2C_2(2\phi+1)^2}\propto \frac{(2\phi+1)^2}{n_{sq}}\to 0,\text{and}\\
\frac{(\sum_k r_k^3)^{1/3}}{(\sum_k\sigma_k^2)^{1/2}}&\le\frac{n_{sq}^{1/3}(2\phi+1)^2}{n_{sq}^{1/2}C_2^{1/2}(2\phi+1)}\propto \frac{2\phi+1}{n_{sq}^{1/6}}\to 0 ,
\end{align*}
where the two limits hold if condition $b)$ holds.
\end{proof}
\bibliographystyle{apalike}      % basic style, author-year citations
\bibliography{ref}   % name your BibTeX data base

%%----------------------------------------------------------------------------------------
%
\end{document}